\documentclass[
  reprint,
  amsmath,
  amssymb,
  aps,
  pra,
  longbibliography
]{revtex4-2}  
\usepackage{amsmath}
\usepackage{amssymb}
\usepackage{amsthm}
\usepackage{mathdots}
\usepackage{bm}
\usepackage{bbm}
\usepackage{graphicx}
\usepackage{epstopdf}
\usepackage{tikz}
\usepackage{amscd}
\usepackage{booktabs}
\usepackage{tabularx}
\usepackage{dcolumn}
\usepackage{multirow}
\usepackage{adjustbox}
\usepackage{enumerate}
\usepackage{changepage}
\usepackage{xcolor}
\usepackage[
colorlinks,
linkcolor = blue,
citecolor = blue,
urlcolor = blue
]{hyperref}

\newtheorem{theorem}{Theorem}

\newtheorem{lemma}{Lemma}

\begin{document}

\preprint{APS/123-QED}
\title{Entanglement thresholds for sequential nonlocality sharing in star networks}

\author{Ming-Xiao Li$^{1,2}$} \thanks{These authors contributed equally}
\author{Yuqi~Li$^{3}$} \thanks{These authors contributed equally}\email{liyuqi@hrbeu.edu.cn}
\author{Rui-Bin Xu$^{1,2}$}\thanks{These authors contributed equally}
\author{Mo-Ran Zhu$^{1,2}$}
\author{Haitao Ma$^{3}$}\email{hmamath@hrbeu.edu.cn}
\author{Chang-Yue Zhang$^{4}$}\email{cyzhang@jnu.edu.cn}
\author{Zhu-Jun Zheng$^{1,2}$}\email{zhengzj@scut.edu.cn}
 \affiliation{ 
\small $^{1}$ School of Mathematics, South China University of Technology, Guangzhou 510640, People's Republic of China\\
\small $^{2}$ Laboratory of Quantum Science and Engineering, South China University of Technology, Guangzhou 510641, People's Republic of China\\
\small $^{3}$ College of Mathematics Science, Harbin Engineering University, Harbin, 150001, Heilongjiang, People's Republic of China\\
\small $^{4}$ Department of Mathematics, College of Information Science and Technology, Jinan University, Guangzhou 510632, People's Republic of China
}
\begin{abstract}
Quantum networks promise scalable distributed quantum information processing. Realizing this potential depends critically on the efficient use of nonclassical resources. Sequential sharing of quantum nonlocality provides a resource-efficient route for recycling entanglement across multiple observers. In network settings, however, such recycling is typically constrained by a trade-off between the number of active branches and the achievable sharing depth. Here we investigate this problem in star networks using probabilistic projective measurements. We establish an explicit entanglement threshold above which nonlocality can be shared across all branches up to any target depth, and below which we determine the optimal trade-off between the number of active branches and the achievable number of sequential rounds. We further analyze the protocol’s robustness to noise and its experimental observability, clarifying its practical feasibility. These results establish quantitative design principles and entanglement requirements for scalable nonlocality recycling in quantum networks.

\end{abstract}
\maketitle
\section{Introduction}
Quantum networks are expected to serve as the physical platform for secure communication~\cite{PhysRevLett.68.3121,PhysRevLett.98.230501,Pironio_2009,PhysRevLett.108.130502,PhysRevLett.110.010503,PhysRevA.89.012301,li2025language}, distributed quantum computation~\cite{PhysRevLett.95.030505,PhysRevA.99.052303,PhysRevLett.125.110505,PhysRevA.104.052404,PhysRevLett.133.120602,zhu2025quest}, and network-enhanced sensing~\cite{degen2017quantum,pirandola2018advances,PhysRevLett.127.060501}. Considerable progress has been achieved in recent years~\cite{PhysRevLett.81.2185,WOS:000340140300015,wehner2018quantum,PhysRevLett.133.230201,PhysRevLett.132.240801,xing2024teleportation}, including long-distance quantum key distribution~\cite{ren2017ground,Liao2017,RevModPhys.94.035001,li2024microsatellitebasedrealtimequantumkey}, metropolitan-scale entanglement distribution~\cite{yin2012quantum,Liu2024,Knaut2024}, and the development of quantum repeater architectures. Yet the transition from proof-of-principle demonstrations to scalable network operation remains limited by a basic experimental bottleneck: High-quality quantum correlations are costly to prepare, fragile under transmission and measurement, and difficult to distribute across many users simultaneously. For this reason, an important practical question is not only how to generate nonlocal correlations~\cite{PhysicsPhysiqueFizika.1.195,PhysRevLett.23.880,RevModPhys.82.665,PhysRevA.85.032119,RevModPhys.86.419,PhysRevA.90.062109,PhysRevLett.115.250401,PhysRevLett.123.140401,PhysRevA.104.042217,PhysRevA.104.042405,Tavakoli_2022,PhysRevLett.128.010403,10849969,PhysRevA.111.012624,Philip2023multipartite} but also how to reuse them efficiently once they have been established.

Quantum nonlocality sharing provides a direct pathway to recycling  quantum resources~\cite{PhysRevLett.125.090401,PhysRevA.106.052412,PhysRevA.105.042436,Zhang2023,PhysRevA.103.022421,PhysRevA.102.032220,PhysRevA.104.L060201,PhysRevA.105.032211,PhysRevA.105.022411,PhysRevA.106.042218,PhysRevA.107.062419,PhysRevResearch.5.013104,PhysRevA.110.012203,PhysRevLett.133.170201,CAI20251,PhysRevA.111.022201,PhysRevA.76.062324,PhysRevA.100.052121}. In this setting, multiple observers can sequentially extract correlations from a single entangled state, while preserving its ability to demonstrate nonlocality for subsequent use. The central issue in these protocols is balancing information extraction against state disturbance. Traditionally, this balance has been achieved through unsharp measurements. More recently, probabilistic projective measurement (PPM) has emerged as a powerful alternative, offering a more experimentally feasible approach. Recent successful experimental demonstrations of sequential sharing in tripartite systems have opened new technological pathways for the distributed utilization of quantum nonlocality~\cite{t57m-xslz}. However, extending this capability to full quantum networks introduces qualitatively new challenges. In a network, early measurements affect not only later observers on the same branch but also the nonlocal resources available across different branches\cite{hu2018observation,Xiao_2024}.

Ideally, one would like to realize sequential sharing across all branches of a network for arbitrarily many rounds, corresponding to a full-depth, full-breadth regime. In practice, however, this ideal scenario is generally constrained by the finite available network entanglement and the difficulty of tailoring measurement settings. Under such constraints, full-network sharing is governed by an intrinsic depth-breadth trade-off: Extending the sharing depth along certain branches typically requires sacrificing the nonlocality on others, and vice versa.  Therefore, two fundamental questions arise for full-network sequential sharing under limited resources: whether there exists a minimum entanglement requirement for sustaining full-network sequential sharing? When this regime is unattainable, what trade-off fundamentally limits the simultaneously sharable network depth and breadth?

In this work, we identify an entanglement threshold that supports ideal full-network nonlocality sharing by introducing a sharing protocol assisted with PPM. For entanglement below this threshold, we characterize the achievable numbers of sharable network branches and sequential sharing rounds, and derive an explicit trade-off between sharing depth and network breadth.  
From a practical perspective, we further present two experimentally relevant theoretical analyses. 
First, we extend the sharing scenario to noisy settings and characterize the effect of noise on the admissible measurement strength, thereby providing guidance for parameter selection under realistic decoherence. 
Second, we compare the detectability of the proposed PPM protocol with that of weak-measurement-based protocols, showing improved detectability for the same number of sharing rounds. Our results clarify the achievable regimes of nonlocality recycling and establish a framework for recycling nonlocal correlations in quantum networks.

\begin{figure}[ptb]
\includegraphics[width=0.48\textwidth]{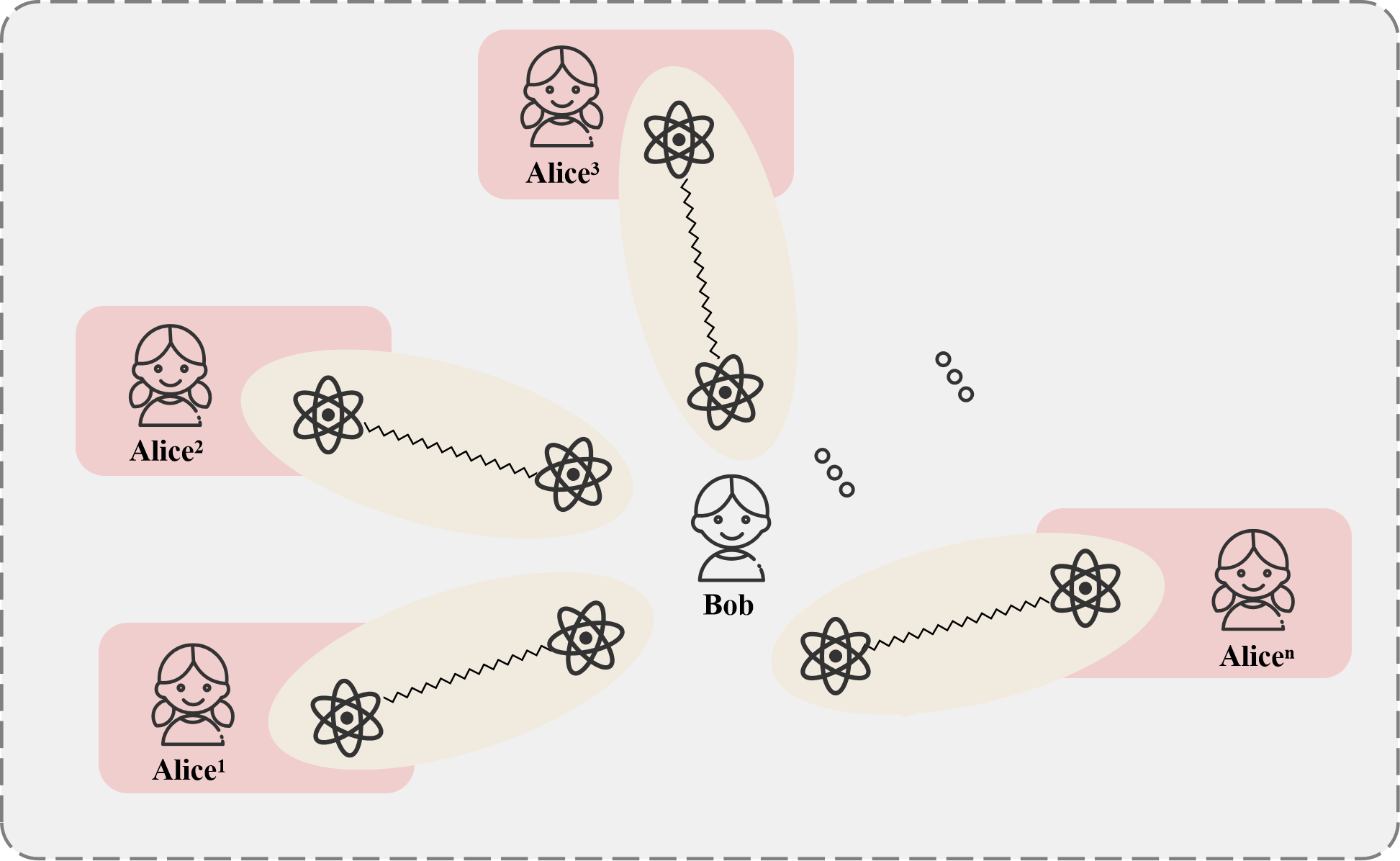}
\caption{\textbf{The $n$-star network.} The central party Bob shares a bipartite entangled state  with each of the $n$ peripheral Alice$^{1}$, Alice$^{2}$, $\cdots$, Alice$^{n}$ via independent sources.}\label{fig:star-network}
\end{figure}


\section{network nonlocality}
Beyond Bell nonlocality in bipartite scenarios, recent research has increasingly focused on understanding quantum correlations in more complex networked structures. Among these, star topology networks have proven to be a particularly powerful framework, facilitating efficient management of quantum correlations, optimized resource distribution, and scalable quantum communication\cite{PhysRevA.90.062109,PhysRevA.104.042217}.

An $n$-star network, as depicted in Fig.~\ref{fig:star-network}, comprises a central node, Bob, connected to $n$ peripheral nodes, Alice$^i$ ($i \in \{1, \dots, n\}$), where Bob shares a bipartite entangled state $\rho_{A^iB}$ with each Alice$^i$ via independent sources. Each party performs local measurements: Alice$^i$ and Bob choose inputs $x_i$ and $y$ from the set $\{0, 1\}$, yielding outcomes $a_i, b \in \{+1, -1\}$, respectively. The resulting joint probability distribution is given by
\begin{align}
&P\left(a_1, \ldots, a_n, b \mid x_1, \ldots, x_n, y\right)\nonumber \\
=&\operatorname{Tr}\left[\left(A_{x_1}^1 \otimes \cdots \otimes A_{x_n}^n \otimes B_y\right)\left(\rho_{A^1 B^1} \otimes \cdots \otimes \rho_{A^n B^n}\right)\right],
\end{align}
where $A_{x_i}^i$ and $B_y$ denote the observables of the respective parties. The network is said to be $n$ local if the probability distribution admits a hidden-variable model.
In this case, it takes the following form:
\begin{equation}
\begin{aligned}
&P  \left(a_1 \ldots a_n b \mid x_1 \ldots x_n y\right) \\
=& \int\left(\prod_{i=1}^n d \lambda_i \mu_i\left(\lambda_i\right) P\left(a_i \mid x_i, \lambda_i\right)\right) P\left(b \mid y, \lambda_1, \ldots, \lambda_n\right),
\end{aligned}
\end{equation}
where $\mu_i(\lambda_i)$ is the distribution of the hidden variable $\lambda_i$. 

To derive testable constraints from this model, it is convenient to introduce the correlation functions:
\begin{equation}\label{exp}
\begin{aligned}
& \left\langle A_{x_1}^1 \ldots A_{x_n}^n B_y\right\rangle \\
 =&\sum_{a_1, \ldots, a_n, b}(-1)^{b+\sum_i a_i} P\left(a_1 \ldots a_n b \mid x_1 \ldots x_n y\right).
\end{aligned}
\end{equation}
For $n$-local probability distributions,  these correlations are bounded by the following inequality~\cite{PhysRevA.90.062109}:
\begin{equation}\label{eq:nonlocal-inequality}
S_n
=
\sqrt[n]{\left|I_n\right|}+\sqrt[n]{\left|J_n\right|} \leqslant 2,
\end{equation}
where
\begin{equation}
\begin{aligned}\label{eq:correlation-functions}
&I_n=\sum_{x_1, \ldots, x_n}\left\langle A_{x_1}^1 \cdots A_{x_n}^n B_0\right\rangle, \\
&J_n= \sum_{x_1, \ldots, x_n}(-1)^{\sum_i x_i}\left\langle A_{x_1}^1 \cdots A_{x_n}^n B_1\right\rangle.
\end{aligned}
\end{equation}
A violation of this classical bound certifies that the $n$-star network exhibits nonlocal correlations. 

To characterize the entanglement of quantum systems, one useful measure is the concurrence~\cite{PhysRevLett.80.2245}. For an arbitrary two-qubit pure state, the concurrence is given by 
\begin{equation}
C\left(|\psi\rangle_{A B}\right)
=
\sqrt{2\left(1-\operatorname{Tr}\left(\rho_A^2\right)\right)}.
\end{equation}
Here, we use $C$ to denote the concurrence of the quantum state $|\psi\rangle_{AB}=\cos\theta|00\rangle+\sin\theta|11\rangle$, which is given by $\sin 2\theta$.
\begin{figure*}[ptb]
\includegraphics[width=0.8\textwidth]{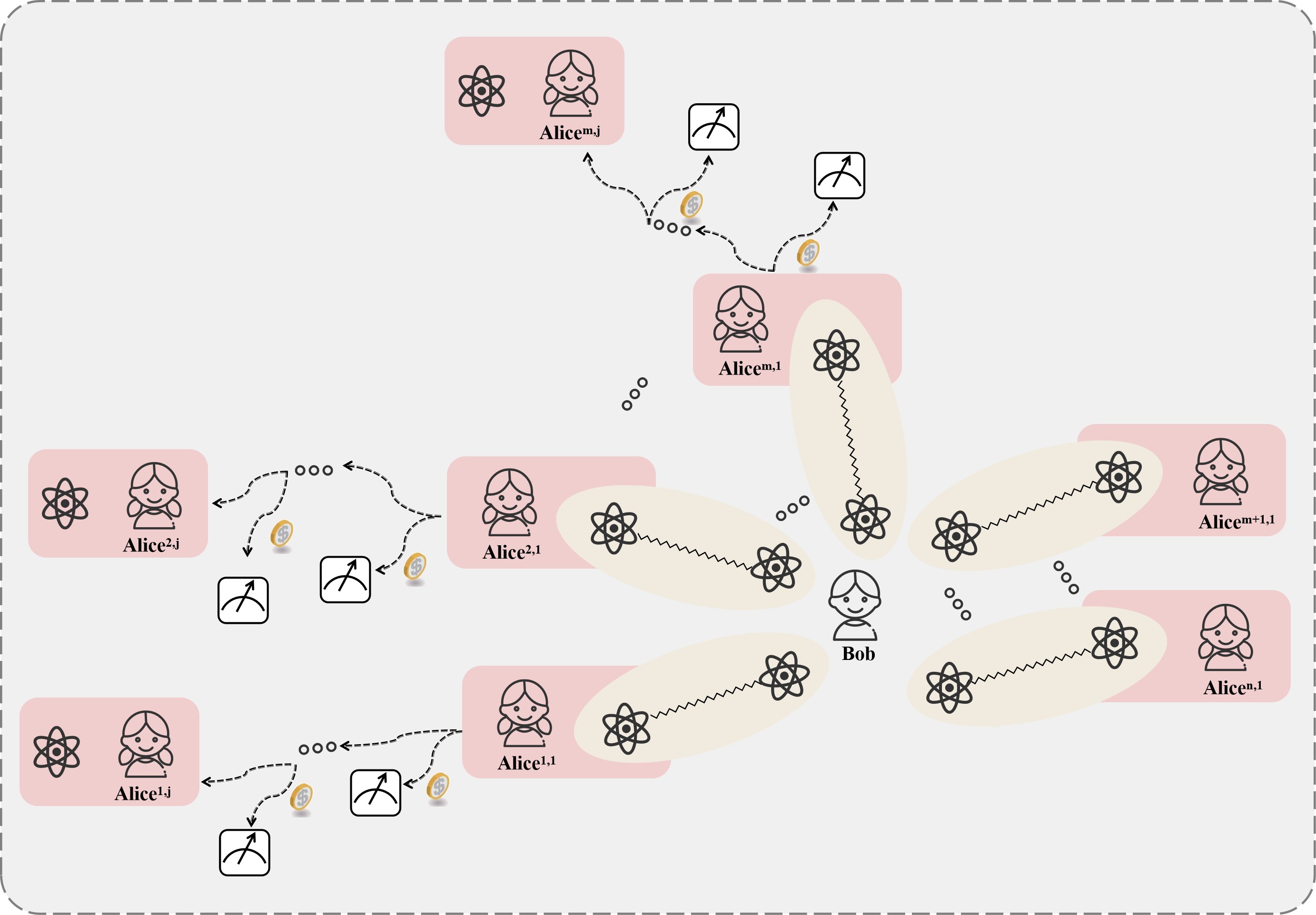}
\caption{\textbf{Sequential nonlocality sharing protocol.} The central node Bob shares an entangled state with each peripheral Alice$^{i,1}$ $(i=1,\dots,n)$. On $m$ branches, sequential observers (Alice$^{i,j}$) perform the PPM strategy: If the input is $x_{i1}=0$, Alice$^{i,j}$ measures $\sigma_x$. If the input is $x_{ij}=1$, she flips a biased classical coin; on ``heads” she measures $\sigma_z$, while on ``tails” she leaves the state unchanged.  Each observer then passes the post-measurement state to the next observer in the same branch, repeating the PPM up to the final layer. The remaining $n-m$ branches are preserved until the final measurement to detect network nonlocality.}\label{fig:nonlocality sharing protocol}
\end{figure*}


\section{Limitations of network nonlocality sharing}
Recycling quantum resources through nonlocality sharing requires controlling the disturbance induced by sequential measurements. As information is extracted sequentially, entanglement degrades and may eventually destroy correlations across the network branches. We therefore begin by asking under what conditions the central node can share nonlocal correlations with all $n$ peripheral nodes simultaneously up to an arbitrary depth $k$.

Answering this question reveals whether the trade-off between sharing depth and network breadth is fundamental and, if not, what entanglement resource is required to sustain multi-round operations across the network. In this section, we establish the quantitative relationship among the target depth $k$, the available entanglement resource $C$, and the maximum number of branches $m$ that can simultaneously sustain nonlocal correlations.

\paragraph*{PPM protocol in the star network}To establish the existence of a valid sharing strategy, we consider a PPM-based sequential protocol to realize nonlocality sharing along $m$ branches of the network.
Each observer Alice$^{i,j}$ on an active branch is equipped with a biased classical coin, which yields ``heads'' with probability $\alpha_{ij}\in(0,1)$ and ``tails'' with probability $1-\alpha_{ij}$. For a first-generation node Alice$^{i,1}$, the measurement process is determined by a classical input bit $x_{i1}\in\{0,1\}$:
\begin{itemize}
    \item if $x_{i1}=0$, Alice$^{i,1}$ performs the projective measurement $A^{i,1}_{0}=\sigma_x$;
    \item if $x_{i1}=1$, Alice$^{i,1}$ flips her coin: Upon obtaining ``heads'', she performs the projective measurement $A^{i,1}_{1}=\sigma_z$, whereas upon obtaining ``tails'', she leaves the subsystem unmeasured.
\end{itemize}

The postmeasurement state is then passed to the next observer on the same branch. The same procedure is then repeated sequentially: Each observer Alice$^{i,j}$ receives the incoming state, performs the operation specified by $x_{ij}$ and $\alpha_{ij}$, and forwards the resulting state to Alice$^{i,j+1}$. This process continues up to the $(k-1)$th generation on each active branch.

Consequently, after $k-1$ sequential rounds on the $m$ active branches, Bob remains correlated with the $k$th-generation observers on those branches, while on the remaining $n-m$ branches he is still correlated with the first-generation observers, as illustrated in Fig.~\ref{fig:nonlocality sharing protocol}. To test the resulting correlations, all Alice nodes that remain correlated with Bob perform a final PPM round, whereas Bob chooses his measurement according to the input $y\in\{0,1\}$, corresponding to the observables
\begin{align}
    B_0
    &=
    (\sin\delta\,\sigma_z+\cos\delta\,\sigma_x)^{\otimes n},
    \nonumber\\
    B_1
    &=
    (-\sin\delta\,\sigma_z+\cos\delta\,\sigma_x)^{\otimes n}.
\end{align}

Therefore, achieving network nonlocality sharing reduces to finding, for a given initial entanglement resource, a compatible choice of the bias sequence $\alpha_{ij}$ and Bob's measurement parameter $\delta$. The parameters $\alpha_{ij}$ determine the measurement strength of the PPM at each round and hence govern the trade-off between information extraction and disturbance of the residual state.  Bob's measurement setting also plays a crucial role in the protocol performance, and its optimization depends on the initial entangled resource.
In what follows, we impose a uniform bias along the $m$ active branches at the $j$th sharing round, namely,
$\alpha_{1j}=\cdots=\alpha_{mj}=\alpha_j$.

To quantify the correlations produced at round $j$, we consider the corresponding network nonlocality expression $S_n^{m,j}$.
\begin{lemma}\label{lem:Snnj-form}
When nonlocality is shared across $m$ branches, the network nonlocality expression in the $j$th round can be written as
\begin{align}
S_n^{m,j}
=
2\,Q_j^{m/n}(\cos\delta+\sin\delta)^{(n-m)/n},
\end{align}
where
\begin{align}
Q_j=\cos\delta\,\sin2\theta\,P_j
+\sin\delta\!\left[
\frac{\alpha_j}{2^{j-1}}+\cos2\theta(1-\alpha_j)
\right],
\end{align}
with $P_1=1$ and $P_j=\prod_{l=1}^{j-1}\frac{2-\alpha_l}{2}$.

\end{lemma}
The details of the derivation are provided in Appendix~\ref{appendixB}.  In particular, for the full-branch case $m=n$, the above expression reduces to $S_n^{n,j}=2Q_j$. Hence, full-network nonlocality sharing up to round $k$ is guaranteed once $Q_j>1$ holds for every $j=1,\dots,k$. We now establish a sufficient condition on the initial entanglement to satisfy this requirement.

\begin{theorem}\label{theorem1}
For any given integer $k$, define $C(k)= 2^{1-k}\sqrt{4^{\,k-1}-1}$. If each source distributes a pure entangled state with concurrence $C \in \bigl(C(k),\,1\bigr]$, then there exists a PPM sequential sharing protocol with a bias sequence $\{\alpha_j\}_{j=1}^{k}$ and a state-dependent measurement parameter $\delta$ for Bob such that, for every round $j\in\{1, \dots, k\}$, 
\begin{equation}\label{eq:K-th nonlocality}
S_n^{n,j}=\sqrt[n]{\left|I_n^{n,j}\right|}+\sqrt[n]{\left|J_n^{n,j}\right|}>2.    
\end{equation}
Hence, network nonlocality sharing along all $n$ branches is achievable up to the $k$-th round.
\end{theorem}
\begin{proof}
We prove the theorem by explicitly constructing a feasible bias sequence $\{\alpha_j\}_{j=1}^k$ together with a suitable choice of Bob's measurement parameter $\delta$. Let each source distribute the pure state
$|\psi\rangle_{A^{i,1}B^i}=\cos\theta|00\rangle+\sin\theta|11\rangle$,
whose concurrence is $C=\sin 2\theta$. 
By Lemma~\ref{lem:Snnj-form}, 
successful $k$-round nonlocality sharing across $n$ branches is equivalent to the condition
\begin{equation}\label{eq:violation-alpha-general}
\alpha_j
>
\frac{2^{j-1}\left(1-\cos\delta\sin2\theta\, P_j-\sin\delta\cos2\theta\right)}
{\sin \delta\left(1-2^{j-1} \cos2\theta\right)},
\end{equation}
for $1\leqslant j\leqslant k$. 

We distinguish two cases based on the source entanglement.

\medskip
\noindent
\textbf{Case 1: $0<\theta<\pi/4$ (non-maximally entangled sources).}
Choose Bob's measurement parameter as
\[
\delta=\frac{\pi}{2}-2\theta.
\]
Then, Eq.~\eqref{eq:violation-alpha-general} reduces to
\begin{equation}\label{eq:violation-alpha-nonmax}
\alpha_j
>
\frac{2^{j-1}\cos^2\delta\, (1-P_j)}
{\sin\delta\left(1-2^{j-1}\sin\delta\right)}.
\end{equation}
Choose $\alpha_1\in(0,1)$ sufficiently small so as to satisfy the first-round violation condition. Then, for $2\le j\le k$, we recursively define
\begin{equation}\label{eq:sequence-probability-nonmax}
\alpha_j
=
(1+\epsilon)
\frac{2^{j-1}\cos^2\delta\, (1-P_j)}
{\sin\delta\left(1-2^{j-1}\sin\delta\right)},
\end{equation}
where $\epsilon>0$ is arbitrarily small, $P_1=1$, and $P_j=\prod_{l=1}^{j-1}\frac{2-\alpha_l}{2}$.  By construction, this sequence strictly satisfies the lower bound in Eq.~\eqref{eq:violation-alpha-nonmax}.

\medskip
\noindent
\textbf{Case 2: $\theta=\pi/4$ (maximally entangled sources).}
This is the limiting case of the choice $\delta=\pi/2-2\theta$, which formally gives $\delta=0$. Since Eq.~\eqref{eq:violation-alpha-general} involves $\sin\delta$ in the denominator, we treat this case by taking $\delta>0$ arbitrarily small and considering the limit $\delta\to0^+$. Using $\cos2\theta=0$ and $\sin2\theta=1$, Eq.~\eqref{eq:violation-alpha-general} becomes
\begin{equation}\label{eq:violation-alpha-max}
\alpha_j
>
\frac{2^{j-1}(1-P_j\cos\delta)}{\sin\delta}.
\end{equation}
In particular, for the first round, $\alpha_1 > (1-\cos\delta)/\sin\delta$. Since the right-hand side tends to $0$ as $\delta\to 0^+$, any fixed $\alpha_1>0$ satisfies this condition for sufficiently small $\delta$. For $2\le j\le k$, we define
\begin{equation}\label{eq:sequence-probability-max}
\alpha_j
=
(1+\epsilon)\frac{2^{j-1}(1-P_j\cos\delta)}{\sin\delta}.
\end{equation}
Again, this choice strictly satisfies Eq.~\eqref{eq:violation-alpha-max}.

\medskip
As detailed in Appendix~\ref{appendixA}, if the concurrence satisfies
\[
C > C(k) = 2^{1-k}\sqrt{4^{\,k-1}-1}.
\] 
the above constructed sequence is strictly increasing and, as $\alpha_1 \to 0$, remains bounded within $(0, 1)$, ensuring it constitutes a valid probability sequence. Therefore, in both cases, there exist a state-dependent measurement parameter $\delta$ and a valid bias sequence $\{\alpha_j\}_{j=1}^k$ such that, for every
 $C \in \bigl(C(k),\,1\bigr]$, network nonlocality sharing along all $n$ branches is achievable up to the $k$th round.
\end{proof}
\begin{figure}[ptb]
\includegraphics[width=0.48\textwidth]{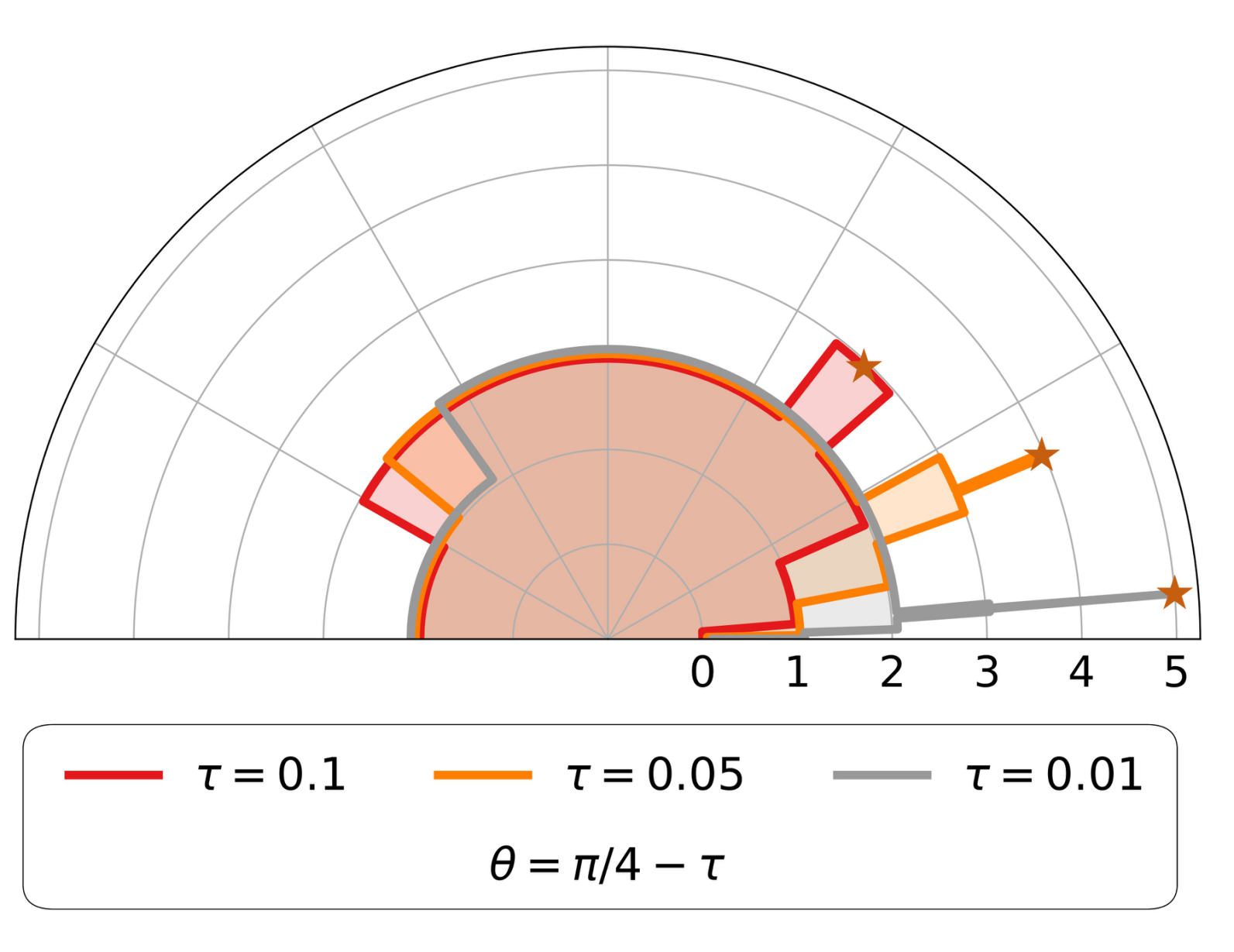}
\caption{\textbf{Achievable sharing rounds.} The figure employs a semicircular coordinate system where the angular position directly represents Bob’s measurement parameter $\delta$, ranging from $\delta = 0$ at the right endpoint to $\delta = \frac{\pi}{4}$ at the left endpoint. The contours are labeled from 1 to 5, indicating the value of sharing rounds. We map the nonlocality-sharing capability for quantum states with different initial parameter $\theta$. The pentagram markers in the figure denote parameter configurations that satisfy the relation $\delta = \frac{\pi}{2}-2\theta$, which is constructed in the proof of Theorem~\ref{theorem1}. Adherence to this relation guarantees, within our protocol, the achievement of the target sequential sharing round $k$ for a given entanglement resource $C>C(k)$.
}\label{fig:versus delta}
\end{figure}

Theorem~\ref{theorem1} establishes a sufficient entanglement condition for $k$-round nonlocality sharing, together with the corresponding measurement settings $\alpha_j$ and $\delta$. This provides a concrete basis for extending multipartite experimental schemes to network nonlocality-sharing scenarios.
In Fig.~\ref{fig:versus delta}, we present the sharing performance for initial states with concurrence $C=\sin 2\theta$. The numerical results corroborate the analytical predictions: For fixed $\theta$, the optimal choice is $\delta=\pi/2-2\theta$, and $k$-round nonlocality sharing is guaranteed whenever $C \in \bigl(C(k),1\bigr]$.

The scenario changes distinctively at the threshold \(C=C(k)\). At this critical point, the sufficient condition of Theorem~\ref{theorem1} no longer guarantees full \(n\)-branch sharing. Substituting
 $C=2^{1-k} \sqrt{4^{k-1}-1}$ into the nonlocality inequality in Eq.~\eqref{eq:K-th nonlocality} for the $k$-th round yields:
\begin{equation}
S_n^{n,k}=2(\cos^2\delta P_k+\sin^2\delta)<2,
\end{equation}
which implies that the sharing fails. This motivates the question of how much of one dimension must be sacrificed to preserve the other. The following theorem establishes an achievable trade-off at \(C=C(k)\).

\begin{theorem}\label{theorem2}
    For any given integer $k$, suppose each source distributes an entangled pure state with concurrence $C=C(k)$. Then there exists a protocol for which the achievable sharing branches $m$ and sharing rounds $j$ satisfy the relation $m+j=n+k-1$.
\end{theorem}
\begin{proof}
We consider the same star-network scenario as in Theorem~\ref{theorem1}, with Bob's measurement parameter set to $
\delta=\frac{\pi}{2}-2\theta.
$

We first consider the full-branch case \(m=n\). Since the fixed resource $C(k)$ is strictly greater than $C(k-1)$ the threshold required for $k-1$ rounds, Theorem~\ref{theorem1} guarantees the existence of a valid probability sequence $\{\alpha_1,\dots,\alpha_{k-1}\}$ that satisfies 
\[
S_n^{n,j}>2,
\quad
j=1,\dots,k-1.
\]
Hence, full-network nonlocality sharing is achievable up to round \(k-1\).

We next turn to the case \(m=n-1\). By Lemma~\ref{lem:Snnj-form}, the condition
$
S_n^{n-1,j}>2
$
is equivalent to
\[
Q_j>(\cos\delta+\sin\delta)^{-1/(n-1)}.
\]
Since
 \((\cos\delta+\sin\delta)^{-1/(n-1)}<1\) for \(\delta\in(0,\pi/2)\),
this requirement is strictly weaker than the full-branch condition \(Q_j>1\). Consequently, the same choice of bias sequence that realizes \(S_n^{n,j}>2\) for \(j=1,\dots,k-1\) also ensures
\[
S_n^{n-1,j}>2,
\quad
j=1,\dots,k-1.
\]

It remains to examine the \(k\)th round. At the threshold \(C=C(k)\), the relation
$
\sin2\theta=2^{1-k}\sqrt{4^{k-1}-1}
$
together with \(\delta=\frac{\pi}{2}-2\theta\) implies
\[
\sin\delta=2^{k-1}\sin^2\delta.
\]
Under this relation, the corresponding quantity in Lemma~\ref{lem:Snnj-form} simplifies to
\[
Q_k=\cos^2\delta\,P_k+\sin^2\delta.
\]
By choosing \(\alpha_1\) sufficiently small, the recursively defined bias sequence yields \(P_k\to1\). Hence,
\[
Q_k\to \cos^2\delta+\sin^2\delta=1.
\]
Substituting this into Lemma~\ref{lem:Snnj-form}, we find
\[
\lim_{\alpha_1\to0} S_n^{n-1,k}
=
2(\cos\delta+\sin\delta)^{1/n}>2,
\quad
\delta\in(0,\pi/2).
\]
Therefore, by continuity, there exists a valid choice of sufficiently small \(\alpha_1>0\) such that
\[
S_n^{n-1,k}>2.
\]
This shows that nonlocality sharing across \(n-1\) branches is achievable up to the \(k\)th round.

Combining the two parts, we conclude that at this threshold, one can trade one branch for one additional round: Full sharing across all \(n\) branches remains achievable up to round \(k-1\), while reducing the breadth to \(n-1\) branches enables sharing up to round \(k\), which satisfy
\[
m+j=n+k-1.
\]
This establishes the claimed trade-off.
\end{proof}

\begin{figure}[ptb]
\includegraphics[width=0.48\textwidth]{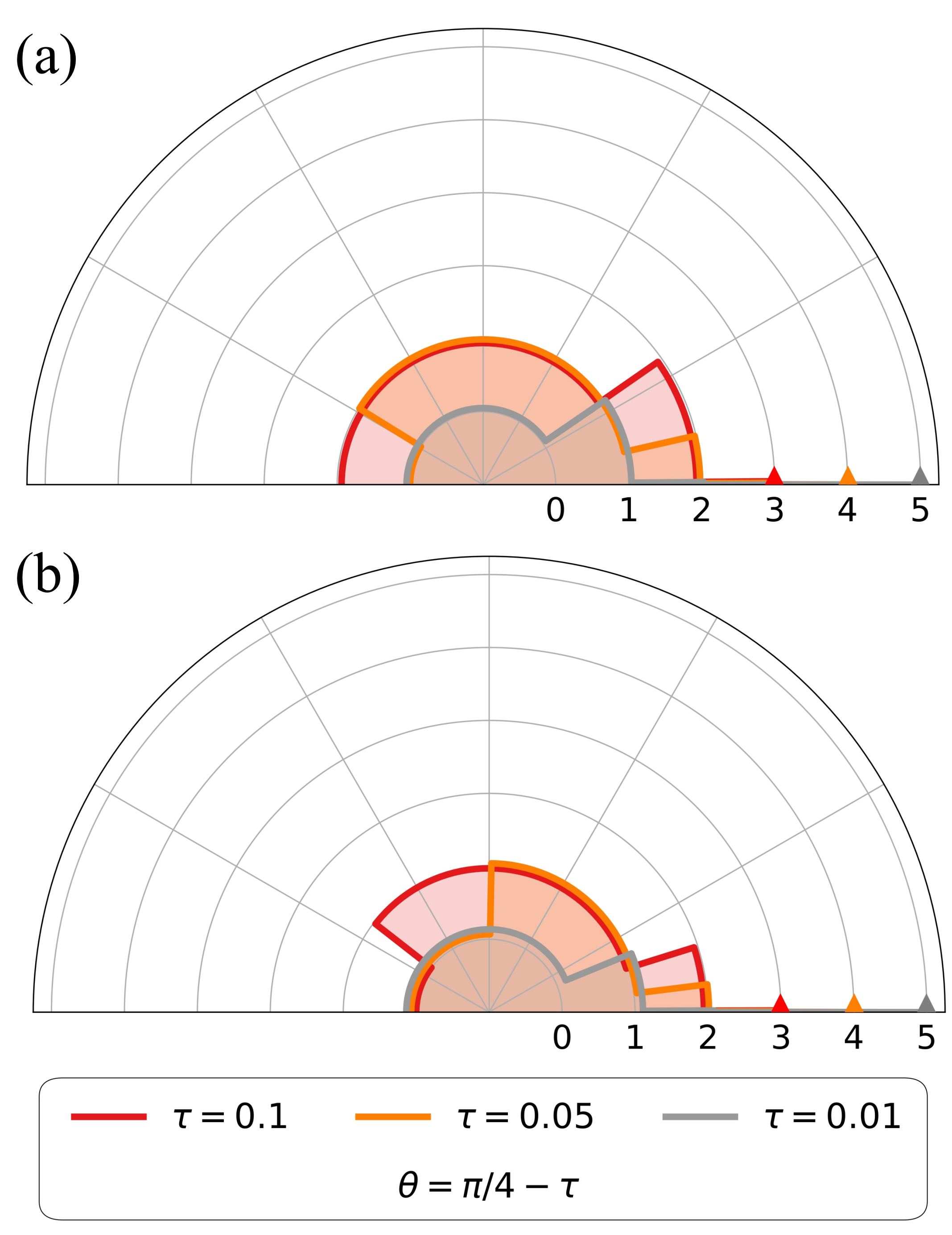}
\caption{\textbf{Sharing capability under noises.} The angular coordinate represents the noise strength, with zero noise at the right endpoint and increasing toward the left. Contour lines labeled from 1 to 5 indicate the maximum number of achievable sharing rounds $k$ for different initial states. Panel (a) covers depolarizing noise, with the noise strength ranging from 0  to 0.1. Panel (b) corresponds to amplitude damping noise, with the damping parameter ranging from 0 to 0.3. The mapping illustrates how the noise level influences the sharing robustness under varying initial entanglement.}\label{fig:noise}
\end{figure}


\section{network nonlocality sharing under noise}
Until now, our analysis has assumed ideal pure entangled states, a common but idealized setting in quantum information processing.  In realistic settings, however, environmental coupling and imperfections in state preparation inevitably introduce noise. For a given initial entanglement resource, the protocol is specified by a compatible choice of the PPM bias parameters $\alpha_{ij}$ and Bob's measurement parameter $\delta$. The central question in the noisy setting is therefore how decoherence modifies this compatibility. To address this question, we derive the nonlocality-violation conditions in the presence of noise and use them to characterize how noise constrains the admissible measurement strengths, modifies the extractable information at each round, and ultimately limits the achievable sequential sharing depth. 

In what follows, we focus on how noise modifies the admissible PPM strengths, while keeping Bob's measurement parameter fixed at the noiseless benchmark value $\delta=\pi/2-2\theta$.  In the noiseless case, this choice maximizes the achievable sharing depth (see Fig.~\ref{fig:versus delta}).
We first consider the case where each initial bipartite entangled source is subjected to independent depolarizing noise. The resulting mixed state shared between Bob and the $i$th Alice becomes
\begin{equation}
\rho_{A^{i,1}B^i}
=
(1-p)|\psi\rangle\langle\psi|+p\frac{\mathbbm{1}}{4},
\quad 
i
=
1\cdots,n,
\end{equation}
with $|\psi\rangle=\cos\theta|00\rangle+\sin\theta|11\rangle$. 
Assuming uniform noise strength $p$ across all branches, the violation condition [Eq.~\eqref{eq:violation-under-depolarizing-noise}] establishes a direct functional relationship between the required measurement probabilities and the noise parameter. Specifically, for the protocol to sustain sequential sharing up to the $k$th round, the probabilities $\{\alpha_j\}$ can be constructed as 
\begin{align}
\alpha_j
=
(1+\epsilon)\frac{2^{j-1}\left((1-p)^{-1}-\sin^22\theta P_j-\cos^22\theta\right)}{\cos2\theta-2^{j-1}\cos^22\theta}.
\end{align}

As a second noise model, we consider amplitude damping. In this setting, the entangled states initially shared between Bob and the Alices are described by
\begin{align}
\rho_{A^{i,1}B^i}
=
\sum_j
\left(I \otimes K_{ij}\right)
|\psi\rangle\langle\psi|
\left(I \otimes K_{ij}\right)^{\dagger},
\end{align}
where the Kraus operators are
\begin{equation}
K_1=\left[\begin{array}{cc}
0 & \sqrt{p_i} \\
0 & 0
\end{array}\right], \qquad K_2=\left[\begin{array}{cc}
1 & 0 \\
0 & \sqrt{1-p_i}
\end{array}\right] .
\end{equation}
Again assuming identical damping strength $p$ across all branches, the violation condition in Eq.~\eqref{eq:violation-under-damping-noise} directly relates the damping parameter to the admissible PPM strengths. For $\epsilon>0$, network nonlocality can be shared up to the $k$th round provided that, for all $1 \le j \le k$, the parameters $\{\alpha_j\}$ are chosen as

\begin{align}
\alpha_j=(1+\epsilon)\frac{2^{j-1}\left(1-\sqrt{1-p}\sin^22\theta P_j-\cos^22\theta\right)}{(1-2p\sin^2\theta)\cos2\theta-2^{j-1}\cos^22\theta}.
\end{align}

For consistency with the noiseless optimization discussed above, we set $\epsilon=10^{-10}$ in the following numerical simulations. Figure~\ref{fig:noise} shows how the achievable number of sharing rounds is affected by the noise strength. In contrast to the noiseless case, we find that although pure states with higher initial entanglement yield a larger sharing depth in the absence of noise, they are also more susceptible to decoherence. Consequently, once noise is present, the admissible PPM strengths and the extractable information per round are more strongly restricted for highly entangled states, so that beyond a certain noise level they can be outperformed by states with lower initial entanglement.

\begin{figure}[ptb]
\includegraphics[width=0.45\textwidth]{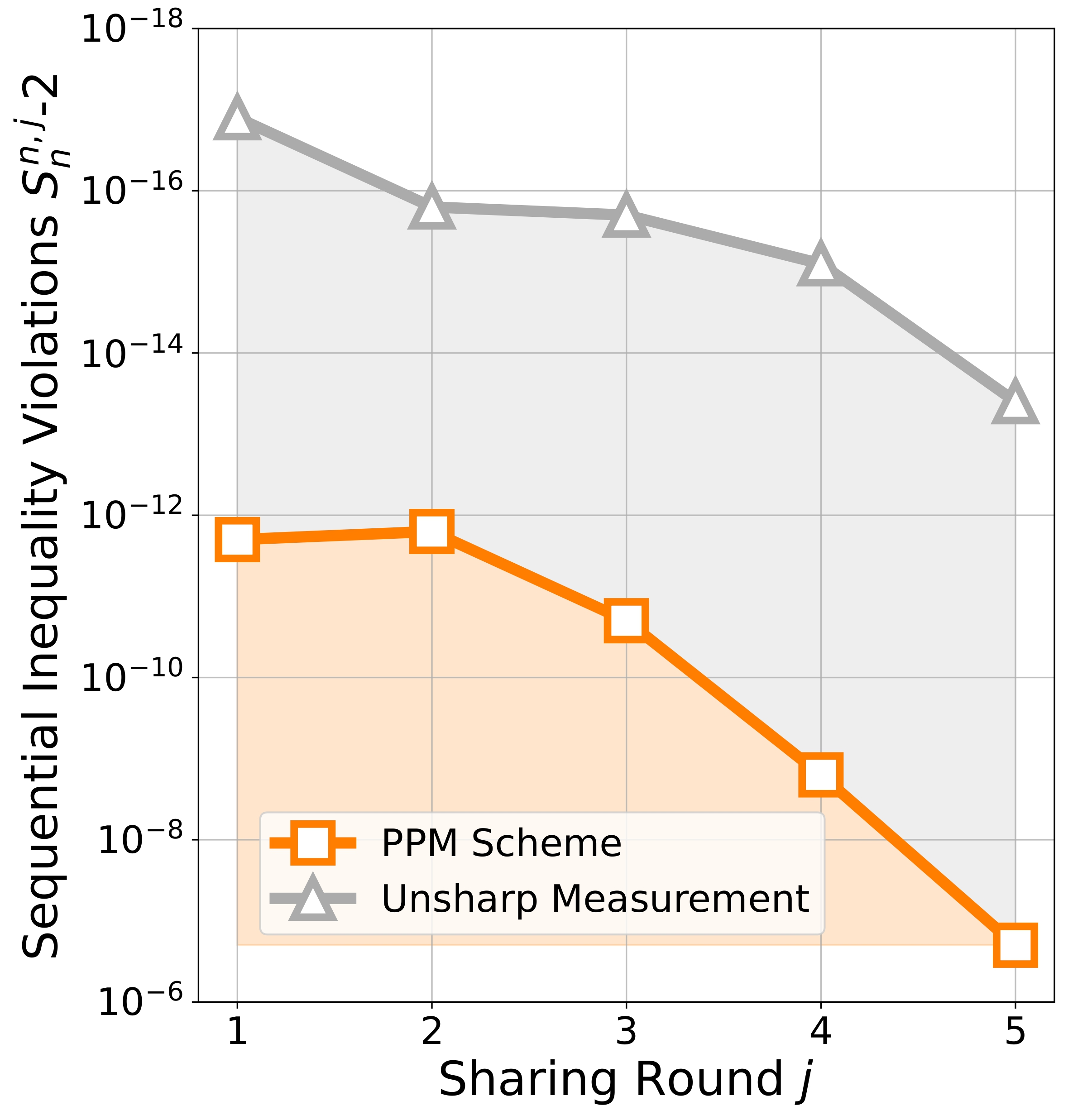}
\caption{\textbf{Comparison with unsharp measurements.} The orange line corresponds to the PPM protocol, while the gray line represents the unsharp measurement protocol. 
For the same sharing depth, the PPM protocol yields consistently stronger nonlocality violations across all rounds, indicating superior experimental detectability.}\label{fig:compare with unsharp}
\end{figure}

\section{Comparison with unsharp measurement}
Beyond noise robustness, another important practical aspect of any sequential sharing protocol is the detectability of the generated nonlocal correlations. Even if many rounds of sharing are theoretically possible, the protocol is of limited experimental value when the resulting nonlocality violations are too weak to be reliably observed. Stronger violations relax the required measurement precision and improve tolerance to imperfections.
Weak measurements constitute a common approach to sequentially recycling quantum correlations, as the reduced interaction strength helps limit state disturbance for later observers~\cite{PhysRevLett.125.090401,PhysRevA.105.042436,PhysRevA.106.052412,Zhang2023}. 
We therefore assess the detectability of the PPM protocol by taking the weak-measurement protocol as a benchmark.

We consider the same star-network setting introduced in the previous section where Bob serves as the central node and simultaneously distributes to each Alice the entangled state $|\psi\rangle_{A^{i,1}B^i}=\cos\theta|00\rangle+\sin\theta|11\rangle$. Each node Alice$^{ij}$ chooses an input $x_{ij}\in\{0,1\}$ and performs either the projective measurement $A_0^{i,j}=\sigma_z$ or the unsharp measurement $A_1^{i,j}=\gamma_{ij}\sigma_x$, with $\gamma_{ij}$ quantifying the degree of disturbance. All Alices within the same generation employ the same unsharpness parameter
$\gamma_{1j}=\cdots=\gamma_{nj}=\gamma_j$. Bob chooses $y\in\{0,1\}$ and performs the projective measurements $B_0 = (\sin\omega\sigma_x + \cos\omega\sigma_z)^{\otimes n}$ and $B_1 = (-\sin\omega\sigma_x + \cos\omega\sigma_z)^{\otimes n}$.
 
Building on the criterion introduced in Ref.~\cite{Zhang2023}, the function $S^{n,j}_n$ can be expressed in terms of the two largest singular values $\lambda_1$ and $\lambda_2$ ($\lambda_1 > \lambda_2$) of the state’s correlation matrix~\cite{HORODECKI1995340}:
\begin{align}\label{eq:weak measurement}
&S^{n,j}_n\nonumber\\
=
&2^{2-j}\left( \gamma_j \sqrt{\lambda_2} \sin \omega + \sqrt{\lambda_1} \cos \omega \prod_{l=1}^{j-1} \left( 1 + \sqrt{1 - \gamma_l^2} \right) \right).
\end{align}
For the quantum state $|\psi\rangle_{A_{i,1}B}$, we have $\lambda_1 = 1$ and $\lambda_2 = \sin 2\theta$. Substituting these values into Eq.~\eqref{eq:weak measurement} 
yields the following condition for violation:
\begin{equation}
\gamma_j > \frac{ 2^{j-1} - \cos \omega \prod_{l=1}^{j-1} \left( 1 + \sqrt{1 - \gamma_l^2} \right) }{ \sqrt{\sin 2\theta}  \sin \omega }.
\end{equation}
To ensure a rigorous comparison, we establish the violation thresholds for both schemes using a unified auxiliary parameter $\epsilon > 0$. The measurement parameters for the unsharp protocol are set as
\begin{equation}\label{eq:gamma}
\gamma_j = (1 + \epsilon)\frac{ 2^{j-1} - \cos \omega \prod_{l=1}^{j-1} \left( 1 + \sqrt{1 - \gamma_l^2} \right) }{ \sqrt{\sin 2\theta} \sin \omega },
\end{equation}
while the corresponding probabilities for the PPM protocol follow from Eq.~\eqref{eq:violation-alpha-general}:
\begin{equation}\label{eq:alpha}
\alpha_j = (1 + \epsilon) \frac{ 2^{j-1} \left( 1 - \sin^2 2\theta P_j - \cos^2 2\theta \right) }{ \cos 2\theta - 2^{j-1} \cos^2 2\theta }.
\end{equation}

With these analytical forms established,  we compare the detectability of the two protocols by evaluating the magnitude of violation per round at the same maximal sharing depth.  For consistency, we fix the initial state parameter to $\theta=\pi/4 - 0.01$ and the auxiliary constant to $\epsilon=10^{-2}$  in both scenarios. Specific measurement parameters are set to $\alpha_1=10^{-10}$ for PPM protocol, while for the weak-measurement protocol, we adopt the setting $\omega=(\pi/4)\times10^{-7}$ from Ref.~\cite{Zhang2023}.
As illustrated in Fig.~\ref{fig:compare with unsharp}, although both approaches successfully achieve a sharing depth of $k=5$, the PPM protocol delivers a consistently stronger violation in each round. This stronger signal enhances experimental visibility and therefore provides a clear practical advantage for verifying network nonlocality.

We emphasize that the present comparison is limited to existing unsharp-measurement-based network-sharing protocols. Whether other weak-measurement variants can offer improved detectability and operational advantages remains an interesting direction for future work. For example, post-selected weak measurement may provide additional advantages, as recent studies suggest its potential for enhancing measurement precision and amplifying weak signals~\cite{PhysRevLett.115.120401,2025EPJP..140..494M}.

\section{conclusion}
In this work, we studied sequential nonlocality sharing in star quantum networks using a PPM protocol. For a target sharing depth $k$, we identified the corresponding entanglement threshold $C(k)$ by explicitly constructing a feasible bias sequence that realizes the protocol. Above this threshold, nonlocality can be shared across all branches up to the desired depth, showing that the usual trade-off between sharing depth and network breadth is not fundamental when sufficient entanglement is available. At the threshold, in contrast, the two objectives can no longer be achieved simultaneously: Retaining all \(n\) branches limits the protocol to round \(k-1\), whereas reaching round \(k\) requires reducing the number of active branches to \(n-1\).  These results clarify how the available entanglement resource constrains sequential nonlocality sharing in star networks.

We then extended our analysis to realistic settings by incorporating depolarizing and amplitude-damping noise. We have provided a complete protocol that supports $k$-round sharing under realistic noise. Our numerical results show that, once realistic imperfections are taken into account, the relation between entanglement and sequential sharing performance becomes more subtle: Although more highly entangled states support greater sharing depth in the ideal case, they can become more fragile under noise and may be outperformed by less entangled states at finite noise strength. We further compared PPM with existing unsharp-measurement-based protocols and found that, for the same maximal sharing depth, PPM yields stronger nonlocality violations at each round, potentially improving experimental detectability.

Several theoretical and experimental directions follow from these findings. On the theoretical side, an important next step is to determine how these results extend beyond star networks, where network topology may introduce new scaling laws and resource trade-offs. Introducing the framework of quantum resource theory into nonlocality sharing could further clarify the limits of quantum correlation reuse.  More broadly, it will be interesting to explore whether other quantum resources, such as coherence, quantum discord, and magic states, can assist sequential nonlocality sharing in distributed settings. On the experimental side, our results provide useful guidance for future implementations, although extending current tripartitle sharing demonstrations to  network architectures will require improved control over photon loss, accumulated imperfections and the realization of biased measurements.

\section{Acknowledgments}
 We would like to thank Chenghong Zhu, Tianfeng Feng, Maosheng Li, and Ya Xi for fruitful discussions.  This research was supported by the National Natural Science Foundation of China under Grant No. 12371458 and No. 12501639, Key Lab of Guangzhou for Quantum Precision Measurement under Grant No. 202201000010, and the Fundamental Research Funds for the Central Universities under Grant No. 21624351. Y. L. and H. M. were supported by the Fundamental Research Funds for the Central Universities (Grant No. 3072025CFJ2407, No. 3072025YC2401, No. 3072025YC2402 and No. 3072025YC2406).

\bibliography{Bib}
\onecolumngrid

\appendix
\section{Construction of sharing protocol}\label{appendixA}
In this section, we present a detailed and rigorous construction of the probability sequences $\{\alpha_j\}_j$ employed in the PPM protocol. These sequences are a central ingredient of our scheme, as they control the trade-off between information extraction and disturbance in each measurement round, and hence determine whether network nonlocality can be shared sequentially over multiple rounds.

\begin{lemma}\label{lemma1}
For any integer $k \geqslant 2$, and pure entangled two-qubit states, characterized by the concurrence, $2^{1-k} \sqrt{4^{k-1}-1}<C=\sin 2 \theta<1$, there exists an increasing sequence $\left\{\alpha_1, \alpha_2, \dots, \alpha_k\right\}$ satisfying
$$
\alpha_1>0\quad \text{and} \quad
0<\frac{2^{j-1} \cos ^2 \delta\left(1-P_j\right)}{\sin \delta\left(1-2^{j-1} \sin \delta\right)}<\alpha_j<1, \quad \forall 2 \leqslant j \leqslant k
$$
where $\delta=\frac{\pi}{2}-2\theta$ and $P_j=\prod_{l=1}^{j-1}\frac{2-\alpha_l}{2}$.
\end{lemma}
The proof of this lemma can be found in Ref.~\cite{PhysRevLett.133.170201}.

\begin{lemma}\label{lemma2}
For any integer $k\geqslant 2$, there exists a parameter $\delta$ and an increasing sequence $\left\{\alpha_1, \alpha_2, \dots, \alpha_k\right\}$ satisfying 
\begin{equation}
 \alpha_1>0 \quad \text{and} \quad  0<\frac{2^{j-1}(1 - \cos\delta \cdot P_j)}{\sin\delta}<\alpha_j<1, \quad \forall 2 \leqslant j \leqslant k,
\end{equation}
where again $P_j=\prod_{l=1}^{j-1}\frac{2-\alpha_l}{2}$.
\end{lemma}
\begin{proof}
We explicitly construct such a sequence. Let $\epsilon>0$ be an arbitrarily small constant and define
\begin{align}\label{eq:sequence-probability-2}
\alpha_1>0, \qquad \alpha_j=(1+\epsilon)\frac{2^{j-1}\bigl(1-\cos\delta\cdot P_j\bigr)}{\sin\delta}, \quad 2\leqslant j\leqslant k.
\end{align}
We first show that the sequence can be made arbitrarily small. Each $\alpha_j$ can be expressed as a polynomial in $\alpha_1$, whose constant term is $\frac{2^{j-1}(1+\epsilon)(1-\cos\delta)}{\sin\delta}$. This term vanishes in the limit $\delta \to 0$, thus guaranteeing that
\begin{align}
    \lim_{\alpha_1 \rightarrow 0^{+}} \alpha_j 
    =
    0\quad 
    \text{for all } j.
\end{align}

Next, we establish the monotonicity of the sequence. Choosing $\alpha_1$ such that
$$\frac{1-\cos\delta}{\sin\delta}<\alpha_1<\frac{2(1-\cos\delta)}{\sin\delta},$$
we obtain for $j=2$,
$$ \frac{\alpha_2}{\alpha_1}=(1+\epsilon)\left(\frac{\cos\delta}{\sin\delta}+\frac{2(1-\cos\delta)}{\alpha_1\sin\delta}\right)>(1+\epsilon)(1+\frac{\cos\delta}{\sin\delta})>1.$$
For the general case $j\geqslant3$, we compute
\begin{align}
\frac{\alpha_j}{\alpha_{j-1}}&=2\times\frac{1-\cos\delta,P_j}{1-\cos\delta,P_{j-1}}=2\times\frac{1-\cos\delta,P_{j-1}\left(1-\frac{\alpha_{j-1}}{2}\right)}{1-\cos\delta,P_{j-1}}=2+\frac{\cos\delta,P_{j-1}\alpha_{j-1}}{1-\cos\delta,P_{j-1}}>2+\frac{2^{j-1}P_{j-1}}{\sin\delta}>1.
\end{align}
Finally, by taking the joint limit $\alpha_1,\delta\to0^+$, we ensure that all $\alpha_j$ lie in the interval $(0,1 )$, while the strict inequalities $\alpha_j>\alpha_{j-1}$ are preserved. This completes the proof.
\end{proof}

\section{Derivation of the network nonlocality inequality under the PPM protocol}
\label{appendixB}

In this section, we present a complete derivation of the network nonlocality inequality associated with the PPM protocol in the $n$-star network. We explicitly compute the average post-measurement state after each sharing round and evaluate the resulting correlators between the central node and the peripheral observers. This appendix provides the full technical justification for the expression of $S_n^{m,j}$ used in the main text. 

\begin{lemma}\label{lemma3}
Let $\rho_{j} (j=2,3,\dots)$ denote the average state after all observers $\text{Alice}^{1,j-1}$, $\text{Alice}^{2,j-1}$, $\cdots$, $\text{Alice}^{m,j-1}$ have performed their probabilistic projective measurements and passed their respective postmeasurement states to the subsequent observers $\text{Alice}^{1,j}$, $\text{Alice}^{2,j}$, $\cdots$, $\text{Alice}^{m,j}$, under the assumptions of uniform measurement bias $\alpha_{1j} = \cdots = \alpha_{mj} = \alpha_j$, $\rho_j$ has the form
\begin{equation}\label{rho}
  \rho_{j}=\sum_{p+q=0}^m\sum_{\substack{\{\beta_1,\cdots,\beta_p \}\\\subset\{1,\cdots,m\}}}\sum_{\substack{\{\gamma_1\cdots \gamma_q\}\subset\\\{1,\cdots,m\}\backslash\{\beta_1,\cdots,\beta_p \}}}\frac{\alpha_{j-1}^{q}(3-\alpha_{j-1})^{m-p-q}}{2^{2m}}\mathcal{A}_0^{\beta_1\cdots\beta_p,j-1}\mathcal{A}_1^{\gamma_1\cdots\gamma_q,j-1}\rho_{j-1} \mathcal{A}_0^{\beta_1\cdots\beta_p,j-1}\mathcal{A}_1^{\gamma_1\cdots\gamma_q,j-1},
         \end{equation}   

where $\mathcal{A}_0^{\beta_1\cdots\beta_p,j-1}=A^{\beta_1,{j-1}}_{0}\otimes\cdots
\otimes A^{\beta_p,{j-1}}_{0}$ and $\mathcal{A}_1^{\gamma_1\cdots\gamma_q,j-1}=A^{\gamma_1,{j-1}}_{1}\otimes\cdots\otimes A^{\gamma_q,{j-1}}_{1}$.
\end{lemma}

\begin{proof}
 For an $n$-star network state $\rho = \bigotimes_{i=1}^n \rho_{A^iB^i}$, where Alice$^1,\dots,$ Alice$^m$ perform projective measurements $A_{\pm|x_i}^i=\frac{1}{2}(\mathbbm{1}\pm A_{x_i}^i),i=1\cdots m$, a direct induction argument shows that averaging over all outcomes yields

\begin{align}\label{eq:average-rho}
\tilde{\rho}&=\sum_{r=0}^{m}\sum_{\substack{\{i_1,\cdots i_r\}\\\subset\{1,\cdots,m\}}}\mathcal{A}_{x_{i_1\cdots x_{i_r}}}^{i_1\cdots i_r}\rho\mathcal{A}_{x_{i_1,\cdots x_{i_r}}}^{i_1\cdots i_r},
\end{align}
where $\mathcal{A}_{x_{i_1\cdots x_{i_r}}}^{i_1\cdots i_r}=A_{x_{i_1}}^{i_1}\otimes A_{x_{i_2}}^{i_2}\otimes\cdots\otimes A_{x_{i_r}}^{i_r}$. This is proved via mathematical induction:
For $k=1$, the statement holds because
\begin{align}
\tilde{\rho}&=\frac{\mathbbm{1}+A^1_{x_1}}{2}\rho\frac{\mathbbm{1}+A^1_{x_1}}{2}+\frac{\mathbbm{1}-A^1_{x_1}}{2}\rho\frac{\mathbbm{1}-A^1_{x_1}}{2}=\frac{1}{2}\left(\rho+A^1_{x_1}\rho A^1_{x_1}\right).
\end{align}  
Assume that the statement holds for $k=m-1$, Then, for $k=m$, we have
\begin{align}
    \tilde{\rho}=&\frac{1}{2^{m-1}}\Bigg[\frac{\mathbbm{1}+A^m_{x_m}}{2}A^1_{x_1}\cdots A^{m-1}_{x_{m-1}}\rho A^1_{x_1}\cdots A^{m-1}_{x_{m-1}}\frac{\mathbbm{1}+A^m_{x_m}}{2}+\frac{\mathbbm{1}-A^m_{x_m}}{2}A^1_{x_1}\cdots A^{m-1}_{x_{mm-1}}\rho A^1_{x_1}\cdots A^{m-1}_{x_{m-1}}\frac{\mathbbm{1}-A^m_{x_m}}{2}\nonumber\\
    &+\frac{\mathbbm{1}+A^m_{x_m}}{2}\sum_{\substack{\{i_1,\cdots, i_{m-2}\}\\\subset\{1,\cdots m-1\}}}A^{i_1}_{x_{i_1}}\cdots A^{i_{m-2}}_{x_{i_{m-2}}}\rho A^{i_1}_{x_{i_1}}\cdots A^{i_{m-2}}_{x_{i_{m-2}}}\frac{\mathbbm{1}+A^m_{x_m}}{2}\nonumber\\
    &+\frac{\mathbbm{1}-A^m_{x_m}}{2}\sum_{\substack{\{i_1,\cdots, i_{m-2}\}\\\subset\{1,\cdots m-1\}}}A^{i_1}_{x_{i_1}}\cdots A^{i_{m-2}}_{x_{i_{m-2}}}\rho A^{i_1}_{x_{i_1}}\cdots A^{i_{m-2}}_{x_{i_{m-2}}}\frac{\mathbbm{1}-A^m_{x_m}}{2}\nonumber\\
   &+\cdots+\frac{\mathbbm{1}+A^m_{x_m}}{2}\sum_{i_1=1}^{m}A^{i_1}_{x_{i_1}}\rho A^{i_1}_{x_{i_1}}\frac{\mathbbm{1}+A^m_{x_m}}{2}+\frac{\mathbbm{1}-A^m_{x_m}}{2}\sum_{i_1=1}^{m}A^{i_1}_{x_{i_1}}\rho A^{i_1}_{x_{i_1}}\frac{\mathbbm{1}-A^m_{x_m}}{2}\nonumber\\
   &+\frac{\mathbbm{1}+A^m_{x_m}}{2}\rho\frac{\mathbbm{1}+A^m_{x_m}}{2}+\frac{\mathbbm{1}-A^m_{x_m}}{2}\rho\frac{\mathbbm{1}-A^m_{x_m}}{2}\Bigg]\nonumber\\
   =&\frac{1}{2^m}\bigg[A^1_{x_1}A^2_{x_2}\cdots A^m_{x_m}\rho A^1_{x_1}A^2_{x_2}\cdots A^m_{x_m}+\sum_{\substack{\{i_1,\cdots, i_{m-1}\}\\\subset\{1,\cdots m\}}}A^{i_1}_{x_{i_1}}\cdots A^{i_{m-1}}_{x_{i_{m-1}}}\rho A^{i_1}_{x_{i_1}}\cdots A^{i_{m-1}}_{x_{i_{m-1}}}+\cdots+\sum_{i_1=1}^{m}A^{i_1}_{x_{i_1}}\rho A^{i_1}_{x_{i_1}}\nonumber+\rho\Bigg].
\end{align}
Thus, by induction, Eq.~\eqref{eq:average-rho} is proved.
Consequently, under PPM protocol the state $\rho_j$ shared by Alice$^{1, j}$, Alice$^{2, j}, \dots$, Alice$^{m, j}, \dots$, Alice$^{m+1,1}, \dots$, Alice$^{n,1}$ and Bob has the form
\begin{align}
\rho_{j}=&\frac{1}{2^m}\sum_{q+p=0}^{m}
\sum_{\substack{\{\beta_1,\cdots,\beta_p \}\\\subset\{1,\cdots,m\}}}\sum_{\substack{\{\gamma_1\cdots \gamma_q\}\subset\\\{1,\cdots,m\}\backslash\{\beta_1,\cdots,\beta_p \}}}\sum_{x_{\beta_s\subset\{0,2\}}}\sum_{x_{\gamma_t\subset\{1,2\}}}\nonumber\\  
&\frac{\alpha_{j-1}^q(1-\alpha_{j-1})^{m-p-q}}{2^{p+q}}   A^{\beta_1,{j-1}}_{x_{\beta_1}}\cdots A^{\beta_p,{j-1}}_{x_{\beta_p}}A^{\gamma_1,{j-1}}_{x_{\gamma_1}}\cdots A^{\gamma_q,{j-1}}_{x_{\gamma_q}}\rho_{j-1} A^{\beta_1,{j-1}}_{x_{\beta_1}}\cdots A^{\beta_p,{j-1}}_{x_{\beta_p}}A^{\gamma_1,{j-1}}_{x_{\gamma_1}}\cdots A^{\gamma_q,{j-1}}_{x_{\gamma_q}},
\end{align}  
where $A^{l_s,j}_2 = A^{g_t,j}_2 = \mathbbm{1}$. 
Based on our PPM protocol, we can assume that among the $m$ systems, $p$ systems perform measurement $A_0$, $q$ systems perform measurement $A_1$, and the remaining systems undergo no measurement. By exhaustively considering all possible measurement configurations, we obtain the first three summation symbols together with the probability $\alpha_{j-1}^q(1-\alpha_{j-1})^{m-p-q}$. Furthermore, according to the structure of the post-measurement quantum state established in Eq.~\eqref{eq:average-rho}, the remaining two summation symbols naturally arise, accompanied by the normalization factor $\frac{1}{2^{p+q}}$. Noting that each term in $\rho_j$ admits the form $$A^{\beta_1,{j-1}}_{0}\cdots A^{\beta_p,{j-1}}_{0}A^{\gamma_1,{j-1}}_{1}\cdots A^{\gamma_q,{j-1}}_{1}\rho_{j-1} A^{\beta_1,{j-1}}_{0}\cdots A^{\beta_p,{j-1}}_{0}A^{\gamma_1,{j-1}}_{1}\cdots A^{\gamma_q,{j-1}}_{1},$$ we can determine the corresponding coefficient for each such term as
\begin{align}
&\sum_{u+v=0}^{m-p-q}\binom{m-p-q}{u}\binom{m-p-q-u}{v}\frac{\alpha_{j-1}^{q+v}(1-\alpha_{j-1})^{m-(p+u)-(q+v)}}{2^{(p+u)+(q+v)}}\nonumber\\
=&\left(\frac{1}{2}\right)^p\left(\frac{\alpha_{j-1}}{2}\right)^q\left(1-\frac{\alpha_{j-1}}{2}+\frac{1}{2}\right)^{m-p-q}\nonumber\\
=&\frac{\alpha_{j-1}^{q}(3-\alpha_{j-1})^{m-p-q}}{2^{m}}.
\end{align}   
Consequently, we recover the expression of $\rho_j$ given in Eq.~\eqref{rho}.
\end{proof}

\begin{lemma}\label{lemma4}
Under our PPM protocol, the network composed of Alice$^{1,j}$, $\cdots$, Alice$^{m,j}$,  Alice$^{m+1,1}$, $\cdots$, Alice$^{n,1}$, and Bob is nonlocal if
\begin{align}
S_{n}^{m,j}=\sqrt[n]{I_n^{m,j}}+\sqrt[n]{J_n^{m,j}}=2\left(\cos\delta\sin2\theta\prod_{l=1}^{j-1}\frac{2-\alpha_l}{2}+\sin\delta\frac{\alpha_j}{2^{j-1}}+\sin\delta\cos2\theta(1-\alpha_j)\right)^\frac{m}{n}(\cos\delta+\sin\delta)^\frac{n-m}{n}>2.
\end{align}
\end{lemma}
\begin{proof}

Under the PPM protocol, the quantum network comprised by Alice$^{1,j}$, $\cdots$, Alice$^{m,j}$,  Alice$^{m+1,1}$, $\cdots$, Alice$^{n,1}$, and Bob exhibits network nonlocality whenever    
\begin{equation}
S_n^{m.,j}=\sqrt[n]{\left|I_n^{m,j}\right|}+\sqrt[n]{\left|J_n^{m,j}\right|} > 2, 
\end{equation}
where the correlators $I_n^{m,j}$ and $J_n^{m,j}$ are defined as 
\begin{align}\label{i}
I_n^{m,j}=&\sum_{\substack{x_{m+1}\cdots,x_n\\\in\{0,1\}}}\sum_{h+f=0}^{m}
\sum_{\substack{\{\xi_1,\cdots,\xi_h \}\\\subset\{1,\cdots,m\}}}\sum_{\substack{\{\zeta_1\cdots \zeta_f\}\subset\\\{1,\cdots,m\}\backslash\{\xi_1,\cdots,\xi_h \}}}\alpha_{j}^f(1-\alpha_{j})^{m-h-f}\text{Tr}\left\{\left(\mathcal{A}_0^{\xi_1\cdots\xi_h,j}\mathcal{A}_1^{\zeta_1\cdots\zeta_f,j}\mathcal{A}_{x_{m+1}\cdots x_n}^{m+1 \cdots n,j}B_0\right)\rho_{j}\right\},
\end{align}
\begin{align}\label{j}
J_n^{m,j}=&\sum_{\substack{x_{m+1}\cdots,x_n\\\in\{0,1\}}}\sum_{h+f=0}^{m}
\sum_{\substack{\{\xi_1,\cdots,\xi_h \}\\\subset\{1,\cdots,m\}}}\sum_{\substack{\{\zeta_1\cdots \zeta_f\}\subset\\\{1,\cdots,m\}\backslash\{\xi_1,\cdots,\xi_h \}}}(-1)^{m-h+\sum_{x_{i}=m+1}^{n}}\alpha_{j}^f(1-\alpha_{j})^{m-h-f}\nonumber\\&\cdot\text{Tr}\left\{\left(\mathcal{A}_0^{\xi_1\cdots\xi_h,j}\mathcal{A}_1^{\zeta_1\cdots\zeta_f,j}\mathcal{A}_{x_{m+1}\cdots x_n}^{m+1 \cdots n,j}B_1\right)\rho_{j}\right\}.
\end{align}
Equations.~\eqref{i} and \eqref{j} performs a summation over all possible configurations of the $m$ systems, where each system undergoes one of three options: an $A_0$ measurement, an $A_1$ measurement, or no measurement. A weight of $\alpha_j$ is assigned to the cases where an $A_1$ measurement is performed, and a weight of $1 -\alpha_j$ is assigned to the cases with no measurement. 

We now proceed to analyze each term in Eqs.~\eqref{i} and \eqref{j}. To simplify notation, we introduce
$$M_{j,\pm}^{h,f}=\sum_{\substack{x_{m+1}\cdots,x_n\\\in\{0,1\}}}(\pm1)^{\sum_{x_{i}=m+1}^{n}}\mathcal{A}_0^{\xi_1\cdots\xi_h,j}\mathcal{A}_1^{\zeta_1\cdots\zeta_f,j}\mathcal{A}_{x_{m+1}\cdots x_n}^{m+1 \cdots n,j}.$$
In the case where, among the $m$ systems, $h$ systems undergo the $A_0$ measurement and $f$ systems undergo the $A_1$ measurement, we have
\begin{align}\label{trace}
&\text{Tr}\left\{\left(M_{j,+}^{h,f}B_0\right)\rho_{j}\right\} \nonumber\\
=&\sum_{p+q=0}^m\sum_{\substack{\{\beta_1,\cdots,\beta_p \}\\\subset\{1,\cdots,m\}}}\sum_{\substack{\{\gamma_1\cdots \gamma_q\}\subset\\\{1,\cdots,m\}\backslash\{\beta_1,\cdots,\beta_p \}}}\frac{\alpha_{j-1}^{q}(3-\alpha_{j-1})^{m-p-q}}{2^m}\text{Tr}\left\{\left(M_{j}^{h,f}B_0\right)\mathcal{A}_0^{p,j-1}\mathcal{A}_1^{q,j-1}\rho_{j-1} \mathcal{A}_0^{p,j-1}\mathcal{A}_1^{q,j-1}\right\}\nonumber\\
=&\sum_{p+q=0}^m\sum_{\substack{w+u=\text{max}\{0,p-(m-h-f)\}\\w\leqslant h,u\leqslant f}}^{\text{min}\{p,h+f\}}\sum_{\substack{v+z=\text{max}\{0,q-(m-h-f-(p-w-u))\}\\v\leqslant h-w,z\leqslant f-u}}^{\text{min}\{q,h++f-w-u\}}(-1)^{u+v}\binom{h}{w}\binom{f}{u}\binom{h-w}{v}\binom{f-u}{z}\nonumber\\
&\binom{m-h-f}{p-w-u}\binom{m-h-f-(p-w-u)}{q-u-z}\frac{\alpha_{j-1}^{q}(3-\alpha_{j-1})^{m-p-q}}{2^{2m}}\text{Tr}\left\{\left(M_{j}^{h,f}B_0\right)\rho_{{j-1}}\right\} .
             \end{align}
\begin{figure}[t]
\centering
\includegraphics[width=0.48\textwidth]{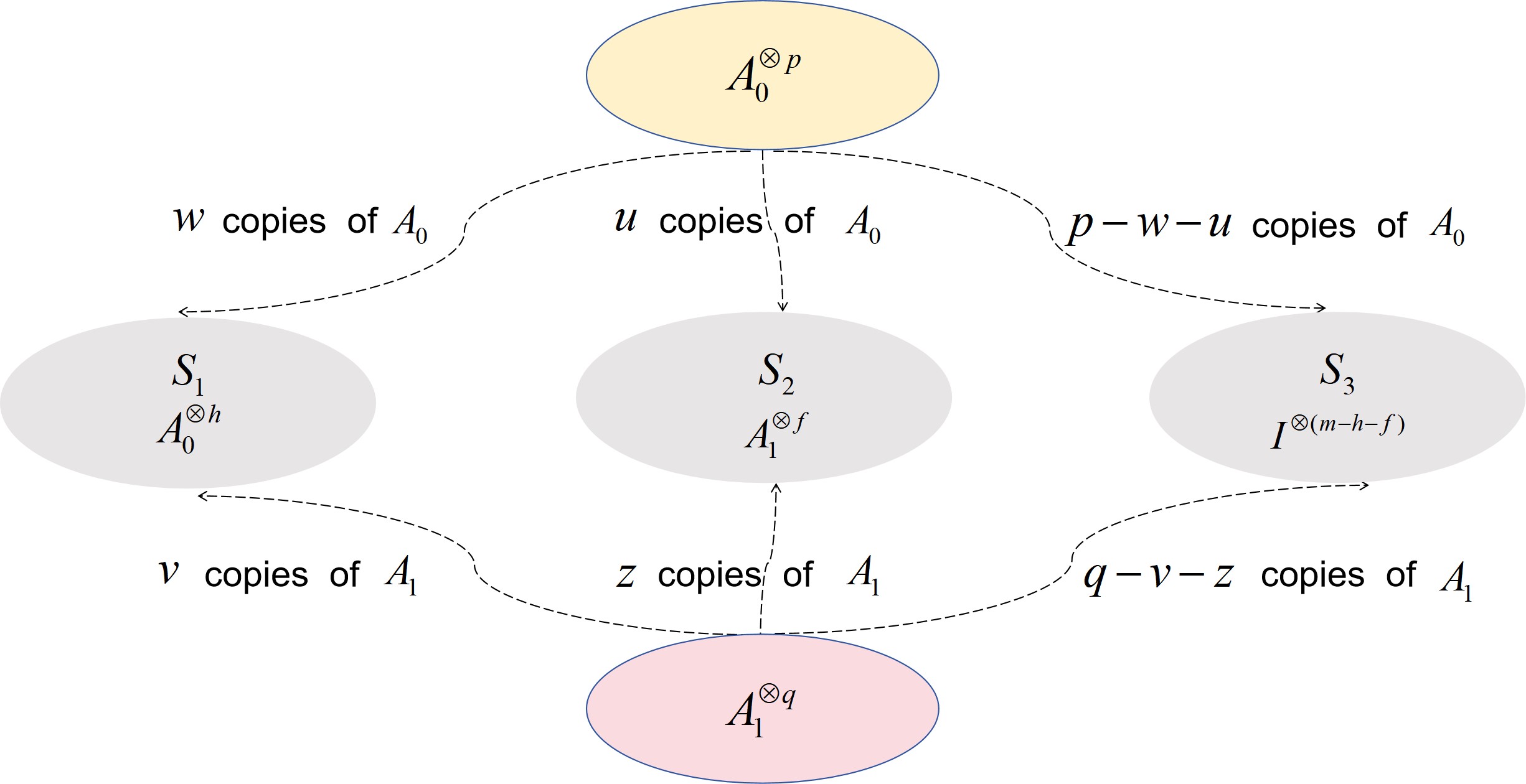}
\caption{Schematic explanation of the derivation of Eq.~\eqref{trace}. The $n$ measurement systems are divided into three groups, as shown in the central ellipses: $h$ systems measured with $A_0$, $f$ with $A_1$, and the remaining $m-h-f$ with no measurement. The top and bottom ellipses represent the $p$ measurements of $A_0$ and $q$ measurements of $A_1$ performed on the state $\rho_{j-1}$, respectively. The arrows indicate how the measurement choices at step $j$ are correlated with the system positions from step $j-1$.}\label{fig:schematic-explanation}
\end{figure}
Using the identities $A_0^2=A_1^2=\mathbbm{1}$ and $A_0A_1=-A_1A_0$, together with the cyclicity
of the trace, the above expression can be reorganized by grouping the $m$ systems into three
subsets (see Fig.~\ref{fig:schematic-explanation}):
\begin{itemize}
    \item $S_1$ : the $h$ systems measured with $A_0$,
    \item $S_2$ : the $f$ systems measured with $A_1$,
    \item $S_3$ : the remaining $m-h-f$ systems.
\end{itemize}
Among the $p$ measurements of $A_0$ acting on $\rho_{j-1}$, suppose $w$ are in $S_1, u$ are in $S_2$, and the remaining $p-w-u$ are in $S_3$. Because each of the $u$ measurements of $A_0$ in $S_2$ anti-commutes with the $A_1$ measurements in that same group, a factor $(-1)^u$ arises. Summing over all such assignments yields the combinatorial factor
$$
(-1)^u\binom{h}{w}\binom{f}{u}\binom{m-h-f}{p-w-u}.
$$
Proceeding analogously for the $q$ measurements of $A_1$ produces an additional factor
$$
(-1)^v\binom{h-w}{v}\binom{f-u}{z}\binom{m-h-f-(p-w-u)}{q-v-z},
$$
where the sign factor $(-1)^v$ originates from anti-commutations with measurements in the respective subgroups.

Carrying out the summations explicitly and applying the binomial theorem, we find
\begin{align}\label{one}
\sum_{q,v,z}(-1)^v
\binom{h-w}{v}\binom{f-u}{z}
\binom{m-h-f-(p-w-u)}{q-v-z}
\frac{\alpha_{j-1}^q(3-\alpha_{j-1})^{m-p-q}}{2^{2m}}
=
\frac{(3-2\alpha_{j-1})^{h-w}3^{m-h-w}}{2^{2m}}.
\end{align}
Substituting this result back into Eq.~\eqref{trace} and summing over $p,w,u$ finally gives
\begin{equation}
\mathrm{Tr}\!\left[\left(M_{j,+}^{h,f}B_0\right)\rho_j\right]
=
\frac{(2-\alpha_j)^h}{2^{h+f}}
\mathrm{Tr}\!\left[\left(M_{j,+}^{h,f}B_0\right)\rho_{j-1}\right].
\end{equation}
Iterating this relation back to the initial state $\rho_1$, we obtain
\begin{align}\label{trace1}
\mathrm{Tr}\!\left[\left(M_{j,+}^{h,f}B_0\right)\rho_j\right]
=&
\prod_{l=1}^{j-1}\frac{(2-\alpha_l)^h}{2^{h+f}}
(\cos\delta\sin2\theta)^h(\sin\delta)^f
(\sin\delta\cos2\theta)^{m-h-f}
(\cos\delta+\sin\delta)^{n-m},
\end{align}
and similarly,
\begin{align}\label{trace2}
\mathrm{Tr}\!\left[\left(M_{j,-}^{h,f}B_1\right)\rho_j\right]
=&
(-1)^{m-h}
\prod_{l=1}^{j-1}\frac{(2-\alpha_l)^h}{2^{h+f}}
(\cos\delta\sin2\theta)^h(\sin\delta)^f
(\sin\delta\cos2\theta)^{m-h-f}
(\cos\delta+\sin\delta)^{n-m}.
\end{align}
Substituting Eqs.~\eqref{trace1} and~\eqref{trace2} into Eqs.~\eqref{i} and~\eqref{j}, respectively,
we finally obtain
\begin{align}
I_n^{m,j}=J_n^{m,j}
=&\sum_{h+f=0}^{m}
\binom{m}{h}\binom{m-h}{f}
\alpha_j^f(1-\alpha_j)^{m-h-f}
\nonumber\\
&\times
\prod_{l=1}^{j-1}\frac{(2-\alpha_l)^h}{2^{h+f}}
(\cos\delta\sin2\theta)^h(\sin\delta)^f
(\sin\delta\cos2\theta)^{m-h-f}
(\cos\delta+\sin\delta)^{n-m}
\nonumber\\
=&
\left(
\cos\delta\sin2\theta\prod_{l=1}^{j-1}\frac{2-\alpha_l}{2}
+\sin\delta\frac{\alpha_j}{2^{j-1}}
+\sin\delta\cos2\theta(1-\alpha_j)
\right)^{m}
(\cos\delta+\sin\delta)^{n-m}.
\end{align}
Therefore, the network is nonlocal whenever
\begin{equation}
S_n^{m,j}
=
2\left(
\cos\delta\sin2\theta\prod_{l=1}^{j-1}\frac{2-\alpha_l}{2}
+\sin\delta\frac{\alpha_j}{2^{j-1}}
+\sin\delta\cos2\theta(1-\alpha_j)
\right)^{\!\frac{m}{n}}
(\cos\delta+\sin\delta)^{\frac{n-m}{n}}
>2.
\end{equation}
\end{proof}

\section{Network nonlocality sharing under noises}
In this section, we extend the analysis of the PPM protocol to realistic scenarios in which
the initially distributed entangled states are affected by environmental noise. We focus on two paradigmatic noise models—depolarizing noise and amplitude damping noise—
and investigate how they modify the conditions under which network nonlocality can be
sequentially shared across all branches of the network.

We first consider the case in which each initial bipartite state is independently subjected
to a depolarizing channel, yielding
\begin{equation}
\rho_{A^{i,1}B^i}=(1-p)|\psi\rangle\langle\psi|+p\frac{\mathbbm{1}}{4},\quad i=1\cdots,n,
\end{equation}
where $|\psi\rangle=\cos\theta|00\rangle+\sin\theta|11\rangle$ and $p\in[0,1]$ denotes the noise
strength. Throughout this section, we assume that the noise parameter $p$ is identical for
all shared states. Using Eqs.~\eqref{trace1} and~\eqref{trace2}, the expectation values entering the correlation
functions can be written as
\begin{align}
\text{Tr}\left\{\left(M_{j,+}^{h,f}B_0\right)\rho_{j}\right\}
=&\prod_{l=1}^{j-1}\frac{(2-\alpha_l)^h}{2^{h+f}}\text{Tr}\left\{\left(M_{j,+}^{h,f}B_0\right)\rho_{1}\right\}\nonumber\\
=&(1-p)^n\prod_{l=1}^{j-1}\frac{(2-\alpha_l)^h}{2^{h+f}}(\cos\delta\sin2\theta)^h(\sin\delta)^f(\sin\delta\cos2\theta)^{n-h-f},
\end{align}
and
\begin{align}
\text{Tr}\left\{\left(M_{j,-}^{h,f}B_0\right)\rho_{j}\right\}
=&\prod_{l=1}^{j-1}\frac{(2-\alpha_l)^h}{2^{h+f}}\text{Tr}\left\{\left(M_{j,+}^{h,f}B_0\right)\rho_{1}\right\}\nonumber\\
=&(1-p)^n(-1)^{n-h}\prod_{l=1}^{j-1}\frac{(2-\alpha_l)^h}{2^{h+f}}(\cos\delta\sin2\theta)^h(\sin\delta)^f(\sin\delta\cos2\theta)^{n-h-f}.
\end{align}
Consequently, the correlators take the form
\begin{align}
I_n^{n,j}=J^{n,j}_n=&\sum_{h+f=0}^{n}\binom{n}{h}\binom{n-h}{f}
\alpha_{j}^f(1-\alpha_j)^{m-h-f}
\nonumber\\&\times\prod_{l=1}^{j-1}\frac{(2-\alpha_l)^h}{2^{h+f}}(\cos\delta\sin2\theta)^h(\sin\delta)^f(\sin\delta\cos2\theta)^{n-h-f}(1-p)^n\nonumber\\
=&\left(\cos\delta\sin2\theta\prod_{l=1}^{j-1}\frac{2-\alpha_l}{2}+\sin\delta\frac{\alpha_j}{2^{j-1}}+\sin\delta\cos2\theta(1-\alpha_j)\right)^n(1-p)^n.\nonumber\\
\end{align}
In particular, for the choice $\delta+2\theta=\pi/2$, network nonlocality can be shared up to
the $k$-th round under depolarizing noise if, for all $1\leqslant j\leqslant k$,
\begin{equation}\label{eq:violation-under-depolarizing-noise}
\alpha_j
>
\frac{2^{j-1}\left((1-p)^{-1}-\sin^22\theta \prod_{l=1}^{j-1}\frac{2-\alpha_l}{2}-\cos^22\theta\right)}{\cos2\theta\left(1-2^{j-1} \cos2\theta\right)},
\end{equation}

We now turn to the amplitude damping noise model. Each initial state undergoes an amplitude damping channel, resulting in
\begin{align}
\rho_{A^{i,1}B^i}&=\sum_j\left(I \otimes K_{j}\right)|\psi\rangle\langle\psi|\left(I \otimes K_{j}\right)^{\dagger}=\left[\begin{array}{cccc}
\cos ^2 \theta & 0 & 0 & \cos \theta \sin \theta \sqrt{1-p} \\
0 & 0 & 0 & 0 \\
0 & 0 & p \sin ^2 \theta & 0 \\
\cos \theta \sin \theta \sqrt{1-p} & 0 & 0 & \left(1-p\right) \sin ^2 \theta
\end{array}\right],\quad i=1\cdots,n,
\end{align}
where
\begin{equation}
K_1=\left[\begin{array}{cc}
0 & \sqrt{p} \\
0 & 0
\end{array}\right], \quad K_2=\left[\begin{array}{cc}
1 & 0 \\
0 & \sqrt{1-p}
\end{array}\right] .
\end{equation}
Assuming again identical noise parameters for all shared states, Eqs.~\eqref{trace1}
and~\eqref{trace2} yield
\begin{align}
\text{Tr}\left\{\left(M_{j,+}^{h,f}B_0\right)\rho_{j}\right\}
=&\prod_{l=1}^{j-1}\frac{(2-\alpha_l)^h}{2^{h+f}}\text{Tr}\left\{\left(M_{j,+}^{h,f}B_0\right)\rho_{1}\right\}\nonumber\\
=&\prod_{l=1}^{j-1}\frac{(2-\alpha_l)^h}{2^{h+f}}(\sqrt{1-p}\cos\delta\sin2\theta)^h((1-2p\sin^2\theta)\sin\delta)^f(\sin\delta\cos2\theta)^{n-h-f},
\end{align}
and
\begin{align}
\text{Tr}\left\{\left(M_{j,-}^{h,f}B_0\right)\rho_{j}\right\}
=&\prod_{l=1}^{j-1}\frac{(2-\alpha_l)^h}{2^{h+f}}\text{Tr}\left\{\left(M_{j,+}^{h,f}B_0\right)\rho_{1}\right\}\nonumber\\
=&(-1)^{n-h}\prod_{l=1}^{j-1}\frac{(2-\alpha_l)^h}{2^{h+f}}(\sqrt{1-p}\cos\delta\sin2\theta)^h((1-2p\sin^2\theta)\sin\delta)^f(\sin\delta\cos2\theta)^{n-h-f}.
\end{align}
Accordingly, the correlatiors become
\begin{align}
I_n^{n,j}=J^{n,j}_n=&\sum_{h+f=0}^{n}\binom{n}{h}\binom{n-h}{f}
\alpha_{j}^f(1-\alpha_j)^{m-h-f}
\nonumber\\&\times\prod_{l=1}^{j-1}\frac{(2-\alpha_l)^h}{2^{h+f}}(\sqrt{1-p}\cos\
delta\sin2\theta)^h((1-2p\sin^2\theta)\sin\delta)^f(\sin\delta\cos2\theta)^{n-h-f}\nonumber\\
=&\left(\sqrt{1-p}\cos\delta\sin2\theta\prod_{l=1}^{j-1}\frac{2-\alpha_l}{2}+(1-2p\sin^2\theta)\sin\delta\frac{\alpha_j}{2^{j-1}}+\sin\delta\cos2\theta(1-\alpha_j)\right)^n.\nonumber\\
\end{align}
For $\delta+2\theta=\pi/2$, network nonlocality can therefore be shared up to the $k$-th round
under amplitude damping noise if, for all $1\leqslant j\leqslant k$,
\begin{align}\label{eq:violation-under-damping-noise}
\alpha_j>\frac{2^{j-1}\left(1-\sqrt{1-p}\sin^22\theta \prod_{l=1}^{j-1}\frac{2-\alpha_l}{2}-\cos^22\theta\right)}{(1-2p\sin^2\theta)\cos2\theta-2^{j-1}\cos^22\theta}.
\end{align}


\end{document}